\documentclass[copyright,creativecommons]{eptcs}
\usepackage{breakurl}               % Not needed if you use pdflatex only.
\usepackage{amssymb}
\usepackage{graphicx}
\usepackage{extarrows}
\usepackage{theorem}
\usepackage[all]{xy}

\newcommand{\eop}{\hspace*{\fill}$\Box$}

\newtheorem{definition}{Definition}

\newtheorem{theorem}[definition]{Theorem}
{\theorembodyfont{\rmfamily}
  \newtheorem{example}[definition]{\it Example}
  \newtheorem{myproof}{\it Proof.}
}

\newenvironment{proof}{\begin{myproof}}{\eop\end{myproof}}

\newcommand{\cA}{\mathcal{A}}
\newcommand{\cF}{\mathcal{F}}
\newcommand{\cG}{\mathcal{G}}
\newcommand{\cI}{\mathcal{I}}
\newcommand{\cP}{\mathcal{P}}
\newcommand{\cS}{\mathcal{S}}

\newcommand{\cV}{\mathcal{V}}

\newcommand{\Var}{{\mathcal{V}ar}}
\newcommand{\FVar}{{\mathit{FV}}}
\newcommand{\Pos}{{\mathcal{P}os}}
\newcommand{\posvar}[1]{{\langle}#1{\rangle}}
\newcommand{\Ran}{{\mathcal{R}an}}
\newcommand{\Terms}{{\mathcal{T}}}

\newcommand{\TwC}[2]{{#1}_{[#2]}}

\newcommand{\Fol}{{\mathcal{F}\!\mathit{ol}}}

\newcommand{\Bothsquares}{\scalebox{.5}[.8]{$\square\!\blacksquare$}}
\newcommand{\Replace}[1]{{#1}|_{\Bothsquares}}
\newcommand{\Subpatterns}{{\mathit{Patterns}_{\lhd}}}
\newcommand{\Inst}{{\mathit{LessGeneralized}}}
\newcommand{\Maximal}[1][]{{\mathit{Least}_{\sqsubseteq #1}}}
\newcommand{\Labels}{{\mathit{Labels}}}
\newcommand{\Initials}{\widetilde{\Labels}}

\renewcommand{\equiv}{=}

\title{On Constructing Constrained Tree Automata Recognizing Ground Instances of Constrained Terms}
\def\titlerunning{On Constructing CTAs Recognizing Ground Instances of
Constrained Terms}
\author{
Naoki Nishida
\institute{Graduate School of Information Science, Nagoya University\\
Furo-cho, Chikusa-ku, 4648603 Nagoya, Japan}
\email{nishida@is.nagoya-u.ac.jp}
\and
Masahiko Sakai
\institute{Graduate School of Information Science, Nagoya University\\
Furo-cho, Chikusa-ku, 4648603 Nagoya, Japan}
\email{sakai@is.nagoya-u.ac.jp}
\and
Yasuhiro Nakano
\institute{Graduate School of Information Science, Nagoya University\\
Furo-cho, Chikusa-ku, 4648603 Nagoya, Japan}
\email{ynakano@trs.cm.is.nagoya-u.ac.jp}
}
\def\authorrunning{N.~Nishida, M.~Sakai, Y.~Nakano}

\begin{document}

%%%%%%%%%%%%%%%%%%%%%%%%%%%%%%%%%%%%%%%%%%%%%%%%%%%
\maketitle
%%%%%%%%%%%%%%%%%%%%%%%%%%%%%%%%%%%%%%%%%%%%%%%%%%%

\begin{abstract}
An inductive theorem proving method for constrained term rewriting
systems, which is based on rewriting induction, needs a decision
procedure for reduction-completeness of constrained terms. In addition,
the sufficient complete property 
of constrained term rewriting systems enables us
to relax the side conditions of some inference rules in the proving
method. These two properties can be reduced to intersection emptiness
problems related to sets of ground instances for constrained terms. 
This paper proposes a method to construct deterministic, complete,
and constraint-complete constrained tree automata recognizing ground
instances of constrained terms.
\end{abstract}

%------------------------------------------------------------------------------
\section{Introduction}
\label{sec:intro}

The \emph{constrained rewriting} in this paper is a computation model
that rewrites a term by applying a constrained rewriting rule if the
term satisfies the constraint on interpretable domains attached to the
rule~\cite{furuichi,BJ08,sakata,SNS11,BJ12}. 
Proposed a \emph{rewriting induction} method on the constrained
rewriting~\cite{sakata,NNKSS11}, a decision procedure of 
\emph{reduction-completeness} of terms must be extended for
\emph{constrained terms}, terms admitting constraints attached, 
in order to apply the method to mechanical inductive theorem proving,
where a term is said to be reduction-complete~\cite{Fri89} 
if any ground instance of the term is a redex. 
If a constrained term rewriting system is terminating and all terms are
reduction-complete, the rewrite system is said to be sufficient complete,
which is useful to relax the application conditions of inference rules
in the above method~\cite{sakata}.
These properties are proved by using tree automata with constraints,
whose rules contain constraints on interpretable subterms.
More precisely, the properties are reducible to the intersection
emptiness problem of ground instances of terms satisfying constraints
attached to the terms. 

This paper proposes a construction method of \emph{constrained tree automata}
that accept ground instances of constrained term (in
Section~\ref{sec:construction}).
Moreover the obtained tree automata have nice properties:
the \emph{constraint-completeness}~\cite{CTA}, completeness and
determinacy,
where the first property is necessary for proving correctness of the
constructed tree automata, and the next two properties contribute avoiding 
size explosion at the construction of product automata. 

%------------------------------------------------------------------------------
\section{Preliminaries}
\label{sec:preliminaries}

In this section, we briefly recall the basic notions of
terms~\cite{TATA}, constraints over predicate logic~\cite{Logic}, and
constrained tree automata~\cite{CTA}. 

%%\paragraph{Terms}

Throughout the paper, we use $\cV$ as a countably infinite set of
\emph{variables}.
For a \emph{signature} $\cF$ (a finite set of function symbols with fixed
arities), the set of \emph{terms} over $\cF$ and $X$ $\subseteq$ $\cV$ is
denoted by $\Terms(\cF,X)$.  
The notation $f/n$ represents the function symbol $f$ with the arity $n$.
%For a function symbol $f$, we denote the arity of $f$ by $\Arity(f)$.
The set $\Terms(\cF,\emptyset)$ of \emph{ground terms} is abbreviated to
$\Terms(\cF)$.
The set of variables appearing in term $t$ is denoted by $\Var(t)$.
A term is called \emph{linear} if any variable appears in the term at
most once. 
%The identity of terms $s$ and $t$ is denoted by $s$ $\equiv$ $t$.
The set of \emph{positions} for term $t$ is denoted by $\Pos(t)$:
\begin{itemize}
 \item $\Pos(x)$ $=$ $\{ \varepsilon \}$ for $x$ $\in$ $\cV$, and
 \item $\Pos(f(t_1,\ldots,t_n))$ $=$ $\{ \varepsilon \} \cup \{ i.\pi
       \mid 1 \leq i \leq n, ~\pi \in \Pos(t_i)\}$ for $f/n$ $\in$
       $\cF$.
\end{itemize}
%The \emph{hole} $\square$ is a special constant not appearing in the initial
%signature.
For terms $t,u$ and a position $\pi$ of $t$, the notation $t[u]_\pi$ denotes
the term obtained from $t$ by replacing the \emph{subterm $t|_\pi$ of $t$
at $\pi$} with $u$.
For a \emph{substitution} $\theta$, we denote the \emph{range} of
$\theta$ by $\Ran(\theta)$. 

%%\paragraph{Constraints}

Let $\cG$ be a signature and $\cP$ be a set of \emph{predicate symbols}
with fixed arities.
%The arity of predicate symbol $p$ is denoted by $\Arity(p)$.
The notation $p/n$ represents the predicate symbol $p$ with the arity
$n$. 
First-order (quantifier-free) \emph{formulas} $\phi$ over $\cG$, $\cP$, 
and $\cV$ are defined in BNF %Backus Naur form
 as follows:\footnote{ It is possible to allow to introduce quantifiers.
To be more precise, introduction of quantifiers to our setting does not
affect the results. 
Though, for the sake of readability, we do not introduce them here. }
\[
  \phi ::= p(t_1,\ldots,t_n) \mid \top \mid \bot \mid (\neg \phi) \mid
 (\phi \vee \phi) \mid (\phi \wedge \phi)
\]
where $p/n$ $\in$ $\cP$ and $t_1,\ldots,t_n$ $\in$ $\Terms(\cG,\cV)$.
We may omit brackets ``$($'', ``$)$'' from formulas as usual. 
The set of first-order formulas over $\cG$, $\cP$, and $X \subseteq \cV$
is denoted by $\Fol(\cG,\cP,X)$.
For a formula $\phi$, the notation $\Var(\phi)$ represents the set of
variables in $\phi$.  
A variable-free formula is called \emph{closed}.
%We also denote the identity of formulas $\phi$ and $\psi$ by $\phi$
%$\equiv$ $\psi$.
A \emph{structure} for $\cG$ and $\cP$ 
%formulas in $\Fol(\cG,\cP,\cV)$ 
is a tuple $\cS$ $=$ $(A,\cI_\cG,\cI_\cP)$ such that the \emph{universe}
$A$ is a non-empty set of concrete values, $\cI_\cG$ and $\cI_\cP$ 
with types $\cG \to \{ f \mid \mbox{$f$ is an $n$-ary function on
$A$}\}$ and $\cP \to 2^{A\times \cdots \times A}$
are interpretations for $\cG$ and $\cP$, resp.:
\begin{itemize}
 \item $\cI_\cG(g)(a_1,\ldots,a_n)$ $\in$ $A$ for $g/n$ $\in$ $\cG$, and
 \item $\cI_\cP(p)$ $\subseteq$ $A^n$ for $p/n$ $\in$ $\cP$.
\end{itemize}
The interpretation of formulas $\phi$ w.r.t.\ $\cS$, denoted by $\cS
\models \phi$, is defined as usual.
%$\cS \models \phi$ if 
% \begin{itemize}
%  \item $\phi$ $\equiv$ $\top$, 
%  \item $\phi$ $\equiv$ $p(t_1,\ldots,t_n)$ and
%	$(\cI_\cG(t_1),\ldots,\cI_\cG(t_n))$ $\in$ $\cI_\cP(p)$, 
%  \item $\phi$ $\equiv$ $\neg \psi$ and $\cS \not\models \psi$,
%  \item $\phi$ $\equiv$ $\phi_1 \vee \phi_2$ and $\cS \models \phi_i$ for
%	some $i$ $\in$ $\{1,2\}$, and
%  \item $\phi$ $\equiv$ $\phi_1 \wedge \phi_2$ and $\cS \models \phi_i$ for
%	all $i$ $\in$ $\{1,2\}$.
% \end{itemize}
We say that a formula $\phi$ \emph{holds} w.r.t.\ $\cS$ if $\cS \models \phi$.
Formulas in $\Fol(\cG,\cP,\cV)$ interpreted by $\cS$ are called
\emph{constraints} (\emph{w.r.t.\ $\cS$}).
%A \emph{constrained term} is a pair $(t,\phi)$ of a term $t$ and a
%constraint $\phi$, written as $\TwC{t}{\phi}$, where $\Var(\phi)
%\subseteq \Var(t)$. 

%\paragraph{Constrained Tree Automata}

In the following, we use $\cF$ and $\cG$ for signatures, $\cP$ for a set of
 predicate symbols, and $\cS$ $=$ $(A,\cI_\cG,\cI_\cP)$ for a structure
 for $\cG$ and $\cP$, without notice. 
%Before the formalization,
Before formalizing constrained tree automata, 
we generalize the interpretation of constraints under terms.
%\begin{definition}%[$\cS,t \models \phi$]
For a sequence $\pi$ of natural numbers, the notation $\posvar{\pi}$ denotes
the special variable related to $\pi$.
We denote the set of such variables by $\posvar{\mathbb{N}^*}$:
$\posvar{\mathbb{N}^*}$ $=$ $\{ \posvar{\pi} \mid \pi \in \mathbb{N}^* \}$
$\subseteq$ $\cV$.
A formula $\phi$ in $\Fol(\cG,\cP,\posvar{\mathbb{N}^*})$ \emph{holds}
 w.r.t.\ a ground term $t$ $\in$ $\Terms(\cF\cup\cG)$ if $\cS,t \models
 \phi$, where $\models$ is inductively defined as follows: 
 \begin{itemize}
  \item $\cS,t \models \top$,
  \item $\cS,t \not\models \bot$,
  \item $\cS,t \models p(t_1,\ldots,t_n)$ if $\pi$ $\in$
	$\Pos(t)$ and $t|_\pi$ $\in$ $\Terms(\cG)$ for all variables
	$\posvar{\pi}$ $\in$ $\bigcup_{i=1}^n \Var(t_i)$, and
	$(\cI_\cG(t_1\theta),\ldots,\cI_\cG(t_n\theta))$ $\in$
	$\cI_\cP(p)$, where $p/n$ $\in$ $\cP$ and $\theta$ is the
	substitution $\{ \posvar{\pi} \mapsto t|_\pi \mid \posvar{\pi}
	\in \bigcup_{i=1}^n \Var(t_i) \}$,
  \item otherwise, $\cS,t \not\models p(t_1,\ldots,t_n)$, and
  \item the relation $\models$ is defined as usual for any Boolean
	connective. % $\neg$, $\vee$, $\wedge$.
%%  \item $\cS,t \models \neg \phi$ if $\cS,t \not\models \phi$,
%%  \item $\cS,t \models \phi_1 \vee \phi_2$ if $\cS,t \models \phi_1$ or
%%	$\cS,t \models \phi_2$, and
%%  \item $\cS,t \models \phi_1 \wedge \phi_2$ if $\cS,t \models \phi_1$
%%	and $\cS,t \models \phi_2$.
 \end{itemize}
We may omit $\cS$ from ``$\cS,t \models \phi$'' if it is clear in context.
Note that $t$ $\models$ $\neg \phi$ does not coincide with $t$
$\not\models$ $\phi$ for every ground term $t$ and constraint $\phi$
(see~\cite{CTA}).
%\end{definition}

By generalizing \emph{AWEDCs}~\cite{TATA}, \emph{constrained tree
automata} are defined as follows~\cite{CTA}~---~note that
\emph{AWEDCs}~\cite{TATA} are CTAs.
\begin{definition}%[constrained tree automaton]
%\label{def:constrained-tree-automata}
%Let $\cF,\cG$ be signatures, $\cP$ be a set of predicate symbols, and
% $\cS$ %$=$ $(A,\cI_\cG,\cI_\cP)$
% be a structure for $\cG$ and $\cP$.
A \emph{constrained tree automata} (CTA) over $\cF$,
 $\cG$, $\cP$, and $\cS$ is a tuple $\cA$ $=$ $(Q,Q_f,\Delta)$ such that
%the only difference from NFTAs is the form of transition rules
 \begin{itemize}
  \item $Q$ is a finite set of \emph{states} (unary symbols), 
  \item $Q_f$ is a finite set of \emph{final} states (i.e., $Q_f$
	$\subseteq$ $Q$), and 
  \item $\Delta$ is a finite set of \emph{constrained transition
	rules}\,% 
\footnote{ We consider transition rules $l \xrightarrow[]{\phi} r$ and
 $l \xrightarrow[]{\psi} r$ to be equivalent if $\Var(\phi)$ $=$
 $\Var(\psi)$ and $\phi$ is semantically equivalent to $\psi$ (i.e.,
 $\phi \leftrightarrow \psi$ is valid w.r.t.\ $\cS$).}
 of the form
 $f(q_1(x_1),\ldots,q_n(x_n)) \xrightarrow[]{\phi} q(f(x_1,\ldots,x_n))$
 $\in$ $\Delta$ where $f/n$ $\in$ 
 $\cF\cup\cG$, $q_1,\ldots,q_n,q$ $\in$ $Q$ and $\phi$ $\in$
 $\Fol(\cG,\cP,\posvar{\mathbb{N}^*})$.  
 \end{itemize}
 We often omit the arguments of states by writing $q$ instead of
 $q(t)$, and then transition rules are written in the form
 $f(q_1,\ldots,q_n) \xrightarrow[]{\phi} q$. 
 We also may omit $\top$ from $f(q_1,\ldots,q_n)  \xrightarrow[]{\top}
 q$. 

The \emph{move relation} $\to_\cA$ is defined as follows:
$t$ $\to_\cA$ $u$ iff $t$ has no nest of state symbols,  $t|_\pi$
 $\equiv$ $f(q_1(t_1),\ldots,q_n(t_n))$, $f/n$ $\in$ $\cF\cup\cG$,
 $t_1,\ldots,t_n$ $\in$ $\Terms(\cF\cup\cG)$,  
$f(q_1,\ldots,q_n) \xrightarrow[]{\phi} q$ $\in$
 $\Delta$, $f(t_1,\ldots,t_n) \models \phi$, 
and $u$ $\equiv$ $t[q(f(t_1,\ldots,t_n))]_\pi$.
\end{definition}
The terminologies of CTAs are defined analogously to those of \emph{tree
 automata}, except for determinism and completeness~---~$\cA$ is called
 \begin{itemize}
  \item \emph{deterministic} if for every ground term
       $t$, there is at most one state $q$ $\in$ $Q$ such that $t$
       $\to^*_\cA$ $q$, and
  \item \emph{complete} if for every ground term $t$,
	there is at least one state $q$ $\in$ $Q$ such that $t$
	$\to^*_\cA$ $q$.
 \end{itemize}
Note that the above definition of determinism and completeness for CTAs
is the same as the definition of the properties for AWEDCs (see
\cite{TATA}). 

\begin{example}\label{ex:CTA}
Let $\cF$ $=$ $\{ \mathsf{f}/2 \}$, $\cG_{\rm int}$ $=$
 $\{ \mathsf{s}/1, \mathsf{p}/1, \mathsf{0}/0 \}$, $\cP$ $=$ $\{ =, \ne, 
 \leq, < \}$, and $\cS_{\rm int}$ be the structure
 $(\mathbb{Z},\cI_{\cG_{\rm int}},\cI_{\cP_{\rm int}})$ for $\cG_{\rm
 int}$ and $\cP_{\rm int}$ such that $\cI_{\cG_{\rm
 int}}(\mathsf{s})(x)$ $=$ $x+1$, $\cI_{\cG_{\rm int}}(\mathsf{p})(x)$
 $=$ $x-1$, $\cI_{\cG_{\rm int}}(\mathsf{0})$ $=$ $0$, and
 $=,\ne,\leq,<$ are interpreted over integers as usual.
 Consider a CTA $\cA_{\rm int}$ $=$ $(\{ q_1, q_2 \},\{ q_2 \},\Delta)$
 over $\cF$, $\cG_{\rm int}$, $\cP_{\rm int}$, and $\cS_{\rm int}$ where 
 \[
 \Delta = 
 \{~~
 \mathsf{0} \to q_1,
 \quad
 \mathsf{s}(q_1) \to q_1, 
 \quad
 \mathsf{p}(q_1) \to q_1, 
 \quad
 \mathsf{f}(q_1,q_1) \xrightarrow[]{\posvar{1.1} \leq \mathsf{p}(\posvar{2})} q_2,
 \quad
 \mathsf{f}(q_1,q_1) \to q_1
 ~~\}
\]
 The term
 $\mathsf{f}(\mathsf{s}(\mathsf{0}),\mathsf{s}(\mathsf{0}))$ is accepted
 by $\cA_{\rm int}$~---~we have that
 $\mathsf{f}(q_1(\mathsf{s}(\mathsf{0})),q_1(\mathsf{s}(\mathsf{0})))
 \to_{\cA_{\rm int}}$ $q_2$ since both 
 $\mathsf{f}(\mathsf{s}(\mathsf{0}),\mathsf{s}(\mathsf{0}))$
 $\to^*_{\cA_{\rm int}}$
 $\mathsf{f}(q_1(\mathsf{s}(\mathsf{0})),q_1(\mathsf{s}(\mathsf{0})))$
 and $\mathsf{f}(\mathsf{s}(\mathsf{0}),\mathsf{s}(\mathsf{0})) \models 
 \posvar{1.1} \leq \mathsf{p}(\posvar{2})$ hold.
 The term $\mathsf{f}(\mathsf{s}(\mathsf{0}),\mathsf{s}(\mathsf{0}))$
 also transitions into $q_1$, and thus, $\cA_{\rm int}$ is not
 deterministic. 
 On the other hand, the term
 $\mathsf{f}(\mathsf{0},\mathsf{s}(\mathsf{0}))$ is not accepted by
 $\cA_{\rm int}$ since
 $\mathsf{f}(\mathsf{0},\mathsf{s}(\mathsf{0})) \not\models \posvar{1.1}
 \leq \mathsf{p}(\posvar{2})$. 
 Note that $\cA_{\rm int}$ recognizes the set of instances of the
 \emph{constrained terms} $\TwC{\mathsf{f}(\mathsf{s}(x),y)}{x \leq \mathsf{p}(y)}$,
 $\TwC{\mathsf{f}(\mathsf{p}(x),y)}{x \leq \mathsf{p}(y)}$, and 
 $\TwC{\mathsf{f}(\mathsf{f}(x,z),y)}{x \leq \mathsf{p}(y)}$,  
 %where $t$ $\in$ $\Terms(\cF\cup\cG_{\rm int})$, 
 i.e.,  
 $L(\cA_{\rm int})$ $=$ 
 $\{ \mathsf{f}(\mathsf{s}(t_1),t_2),
 ~
 \mathsf{f}(\mathsf{p}(t_1),t_2),
 ~
 \mathsf{f}(\mathsf{f}(t_1,t),t_2)
 \mid t \in \Terms(\cF\cup\cG_{\rm int}),
 ~t_1,t_2 \in \Terms(\cG_{\rm int}),~\cI_{\cG_{\rm int}}(t_1) \leq
 \cI_{\cG_{\rm int}}(\mathsf{p}(t_2)) \}$ (see
 Definition~\ref{def:constrained-term}).

 Consider the CTA $\cA_{\rm int}'$ $=$ $(\{ q_1, q_2 \},\{ q_2
 \},\Delta')$ obtained from 
 $\cA_{\rm int}$ by replacing $\posvar{1.1} \leq \mathsf{p}(\posvar{2})$
 with $\neg (\mathsf{p}(\posvar{2}) < \posvar{1.1})$:
 \[
 \Delta' = 
 \{~~
 \mathsf{0} \to q_1,
 \quad
 \mathsf{s}(q_1) \to q_1, 
 \quad
 \mathsf{p}(q_1) \to q_1, 
 \quad
 \mathsf{f}(q_1,q_1) \xrightarrow[]{\neg (\mathsf{p}(\posvar{2}) < \posvar{1.1})} q_2,
 \quad
 \mathsf{f}(q_1,q_1) \to q_1
 ~~\}
\]
 The constraints $x \leq \mathsf{p}(y)$ and $\neg (\mathsf{p}(y) < x)$ are 
 semantically equivalent, i.e., for all terms $t_1,t_2$ $\in$
 $\Terms(\cG_{\rm int})$, $\cS_{\rm int} \models t_1 \leq \mathsf{p}(t_2)$
 iff $\cS_{\rm int} \models \neg (\mathsf{p}(t_2) < t_1)$. 
 However, this is not the case for similar constraints over fixed terms,
 e.g., $\mathsf{f}(\mathsf{0},\mathsf{s}(\mathsf{0})) \models
 \posvar{1.1} \leq \mathsf{p}(\posvar{2})$ does not hold, but
 $\mathsf{f}(\mathsf{0},\mathsf{s}(\mathsf{0})) \models 
 \neg (\mathsf{p}(\posvar{2}) < \posvar{1.1})$ holds.
 Thus, $\mathsf{f}(\mathsf{0},\mathsf{s}(\mathsf{0}))$ is accepted by
 $\cA_{\rm int}'$, and hence $L(\cA_{\rm int})$ $\ne$ $L(\cA_{\rm
 int}')$. 
\end{example}

A CTA $\cA$ $=$ $(Q,Q_f,\Delta)$ is called
\emph{constraint-complete}~\cite{CTA} if for every ground term $t$ $\in$
$\Terms(\cF\cup\cG)$ and all transition rules 
$f(q_1,\ldots,q_n) \xrightarrow[]{\phi} q$ $\in$ $\Delta$ with $t$ $\equiv$
$f(t_1,\ldots,t_n)$ $\to^*_\cA$ $f(q_1,\ldots,q_n)$, 
% and $t \models \phi$, 
we have that $\pi$ $\in$ $\Pos(t)$ and $t|_\pi$ $\in$ $\Terms(\cG)$
for all variables $\posvar{\pi}$ in $\phi$.
Note that every CTA can be transformed into an equivalent,
deterministic, complete, and constraint-complete CTA~\cite{CTA}. 
Note also that completeness and constraint-completeness are different notions. 

%------------------------------------------------------------------------------
\section{Recognizing Ground Instances of Constrained Terms}
\label{sec:construction}

This section defines \emph{constrained terms} and their instances,
and then proposes a method for constructing a CTA recognizing the set of
ground instances for a given set of constrained terms. 
%Moreover, we describe properties of the resulted CTA.
The method is based on the construction of a \emph{tree automaton}
recognizing the set of ground instances for unconstrained terms, which
is complete and deterministic~\cite[Exercise~1.9]{TATA}.
Accessibility of states is in general undecidable, and thus, it is
difficult to get rid of inaccessible states which affect
the \emph{intersection emptiness problem}. 
In this sense, it is worth developing a construction method that
introduces inaccessible states as little as possible. 
\begin{definition}\label{def:constrained-term}
A \emph{constrained term} is a pair $(t,\phi)$, written as
$\TwC{t}{\phi}$, of a linear term $t$ 
 $\in$ $\Terms(\cF\cup\cG,\cV)$ and a formula $\phi$ $\in$
 $\Fol(\cG,\cP,\Var(t))$. 
The set of \emph{ground instances} of a constrained term 
$\TwC{t}{\phi}$, denoted by $G(\TwC{t}{\phi})$, is defined
 as follows:
$G(\TwC{t}{\phi})$ $=$
$\{ t\theta \in \Terms(\cF\cup\cG) 
 \mid \Ran(\theta|_{\FVar(\phi)}) \subseteq \Terms(\cG),
 ~ \cS \models \phi\theta \}$.
%The set of ground instances is extended to the set of constrained terms
% as follows:
The argument of $G$ is naturally extended to sets 
$ %\[
 G(T) = \bigcup_{\TwC{t}{\phi} \in T} G(\TwC{t}{\phi})
$. %\]
\end{definition}

To deal with constrained terms, we consider \emph{constrained
patterns}. 
We introduce wildcard symbols $\square$ and $\blacksquare$
to denote arbitrary interpretable and un-interpretable terms
resp.
In the following, we denote
$\Terms(\cF\cup\cG\cup\{\square,\blacksquare\})$ (the set of patterns)
by $\Terms_{\square,\blacksquare}$, 
$\Terms(\cG\cup\{\square\})$ (the set of interpretable patterns) by
$\Terms_\square$, and
$\Terms(\cF\cup\cG\cup\{\square,\blacksquare\})\setminus\Terms(\cG\cup\{\square\})$
(the set of un-interpretable patterns) by $\Terms_\blacksquare$. 
\begin{definition}
% Let $\square$ and $\blacksquare$ be constants not appeared in
% $\cF\cup\cG$. 
 A pair of a ground term $u$ $\in$
 $\Terms_{\square,\blacksquare}$ and a formula $\phi \in
 \Fol(\cG,\cP,\posvar{\mathbb{N}^+})$ is called a \emph{constrained
 pattern} if $\pi \in \Pos(t)$ and $t|_\pi$ $\in$ $\Terms_\square$ for
 each variable $\posvar{\pi}$ that occurs in $\phi$. 

 For a constrained term $\TwC{t}{\phi}$, $\Replace{(\TwC{t}{\phi})}$
 denotes the set of constrained terms $\TwC{t\theta}{\phi\sigma}$ where
\begin{itemize}
%  \item $\tilde{t}$ is a term obtained from $t$ by replacing with
%    $\square$ each occurrence in $t$ of variables that are free in
%    $\phi$ and replacing with $\square$ or $\blacksquare$ each
%    occurrence in $t$ of the other variables,
 \item $\theta$ $=$ $\{ x \mapsto \square \mid x \in \Var(t) \cap
       \FVar(\phi) \} \cup 
       \bigcup_{y \in \Var(t) \setminus \FVar(\phi)}
            \{ y \mapsto v \mid \text{$v$ is either $\square$ or $\blacksquare$} \}$, and
 \item $\sigma$ $=$ $\{ x \mapsto \posvar{\pi} \mid x \in
       \Var(t)\cap\FVar(\phi),~t|_\pi = x\}$.%
       \footnote{ $\pi$ is unique since $t$ is linear. }
\end{itemize}
We extend the domain of $\Replace{(\cdot)}$ to sets of
 constrained terms:
$\Replace{T}$ $\!=\!$
 $\bigcup_{\TwC{t}{\phi}\in T}
 \Replace{(\TwC{t}{\phi})}$.
\end{definition}
% Note that $\Replace{(\TwC{t}{\phi})}$ is a set of constrained patterns. 
 Roughly speaking, a constrained pattern $\TwC{u}{\phi}$
 represents a set of terms obtained by replacing $\square$ and
 $\blacksquare$ in $u$ by interpretable and un-interpretable terms,
 resp., such that the constraint obtained by the corresponding
 replacement holds.
\begin{example}
 \label{ex:replace}
Let $\cF$ $=$ $\{ \mathsf{g}/2\}$, $\cG$ $=$ $\{
 \mathsf{0}/0,\mathsf{s}/1\}$, 
and $\cP$ $=$ $\{ \leq,\geq \}$.
Let symbols $\mathsf{0}$, $\mathsf{s}$, $\leq$, and $\geq$ be 
interpreted by $\cS$ as zero function and successor function, 
less-or-equal relation, and greater-or-equal relation, resp.
 For $T$ $=$ $\{ \TwC{\mathsf{g}(x,y)}{x \leq \mathsf{0}}, 
        \TwC{\mathsf{g}(\mathsf{s}(x),y)}{x \geq \mathsf{0}} \}$,
\[
 \Replace{T} =
 \{~
    \TwC{\mathsf{g}(\square,\square)}{\posvar{1}\leq \mathsf{0}}, 
    \ \TwC{\mathsf{g}(\square,\blacksquare)}{\posvar{1} \leq \mathsf{0}},
    \ \TwC{\mathsf{g}(\mathsf{s}(\square),\square)}
          {\posvar{1.1} \geq \mathsf{0}}, 
    \ \TwC{\mathsf{g}(\mathsf{s}(\square),\blacksquare)}
          {\posvar{1.1} \geq \mathsf{0}}
 ~\}
\]
\end{example}

Next, we define a function to augment their subterms to constrained patterns. 
\begin{definition}
For a set $U$ of constrained patterns, 
the set $\Subpatterns(U)$ of proper subterms for constrained patterns in
 $U$ is defined as follows: 
\[
 \Subpatterns(U) =
 \{ u|_\pi \mid \TwC{u}{\phi} \in U,
 ~ \pi \in \Pos(u) \setminus \{ \varepsilon \}, ~ u|_\pi \not\in \{\square,\blacksquare\} \}
\]
\end{definition}
\begin{example}
For $T$ in Example~\ref{ex:replace}, 
$\Subpatterns(\Replace{T})$
$=$ $\{ \mathsf{s}(\square) \}$.
\end{example}

We define a quasi-order over constrained patterns 
that represents an approximation relation.
% in other to compare patterns w.r.t.\ concreteness. 
\begin{definition}
 \label{sqsubseteq}
 A quasi-order $\sqsubseteq$ over terms in
 $\Terms_{\square,\blacksquare}$ is inductively defined as
 follows: 
 \begin{itemize}
  \item $\square$ $\sqsubseteq$ $u$ for $u$ $\in$
	$\Terms_\square$,
  \item $\blacksquare$ $\sqsubseteq$ $u$ for $u$ $\in$
	$\Terms_\blacksquare$, and
  \item $f(u_1,\ldots,u_n)$ $\sqsubseteq$ $f(u'_1,\ldots,u'_n)$ if $u_i$
       $\sqsubseteq$ $u'_i$ for all $1$ $\leq$ $i$ $\leq$ $n$.
 \end{itemize}
%We write $u$ $\bumpeq$ $u'$ if $u$ $\sqsubseteq$ $u'$ and $u'$
% $\sqsubseteq$ $u$, and $u$ $\sqsubset$ $u'$ if $u$ $\sqsubseteq$ $u'$
% and $u'$ $\not\sqsubseteq$ $u$. 

Abusing notations, we also define a quasi-order $\sqsubseteq$ over
 formulas in $\Fol(\cG,\cP,X)$ as follows:
$\phi$ $\sqsubseteq$ $\phi'$
iff $\FVar(\phi)$ $\subseteq$ $\FVar(\phi')$ and $\phi' \Rightarrow
 \phi$ is valid w.r.t.\ $\cS$.
% We define $\bumpeq$ and $\sqsubset$ over formulas as well as the case of
%  patterns. 

A quasi-order $\sqsubseteq$ over constrained patterns is defined as
 follows:
 $\TwC{u}{\phi}$ $\sqsubseteq$ $\TwC{u'}{\phi'}$
 iff $u$ $\sqsubseteq$ $u'$ and $\phi$ $\sqsubseteq$ $\phi'$.
For constrained patterns $\TwC{u}{\phi}$ and $\TwC{u'}{\phi'}$, we say
 that $\TwC{u}{\phi}$ is \emph{more general} than $\TwC{u'}{\phi'}$
 ($\TwC{u'}{\phi'}$ is \emph{less general} than $\TwC{u}{\phi}$) if
 $\TwC{u}{\phi}$ $\sqsubseteq$ $\TwC{u'}{\phi'}$. 

We use $\bumpeq$ for equality part of $\sqsubseteq$, and
$\sqsubset$ for strict part of $\sqsubseteq$.
\end{definition}

Next, to compute more concrete patterns, we define an operation $\sqcap$
for constrained patterns.
\begin{definition}
We define a binary operator $\sqcap$ over
 $\Terms_{\square,\blacksquare}$ inductively as follows:
\begin{itemize}
 \item $\square$ $\sqcap$ $u$ $=$ $u$ $\sqcap$ $\square$ $=$ $u$ for $u$
       $\in$ $\Terms_\square$, 
 \item $\blacksquare$ $\sqcap$ $u$ $=$ $u$ $\sqcap$ $\blacksquare$ $=$
       $u$ for $u$ $\in$ $\Terms_\blacksquare$, and
 \item $f(u_1,\ldots,u_n)$ $\sqcap$ $f(u'_1,\ldots,u'_n)$ $=$ $f(u_1
       \sqcap u'_1,\ldots,u_n \sqcap u'_n)$. 
\end{itemize}
\end{definition}

In the following, we define the set of constrained patterns used for
labels of states.
\begin{definition}
For a set $U$ of constrained patterns, $\Inst(U)$ is the smallest set
 satisfying the following:
 \begin{itemize}
  \item $\{\TwC{\square}{\top}, \TwC{\blacksquare}{\top} \}
	\cup \{ \TwC{u}{\top} \mid u \in \Subpatterns(U) \}$
	$\subseteq$ $\Inst(U)$,
  \item $\{ \TwC{u}{\phi'} \mid \TwC{u}{\phi} \in U,
	~
	\phi' \in \{ \phi, \neg \phi \},
	~
	\mbox{$\phi'$ is satisfiable w.r.t.\ $\cS$} \}$
	$\subseteq$ $\Inst(U)$,
	and
  \item $\TwC{(u \sqcap u')}{\phi\wedge\phi'}$ if
	$\TwC{u}{\phi}$, $\TwC{u'}{\phi'}$ $\in$ $\Inst(U)$
	and $\phi\wedge\phi'$ is satisfiable w.r.t.\ $\cS$.
 \end{itemize}
Note that we do not distinguish terms that enjoy $\bumpeq$, and that
 $\Inst(U)$ is finite up to $\bumpeq$. 
\end{definition}
\begin{example}
 \label{ex:Inst}
Consider $(\cF,\cG,\cP,\cS)$ and $T$ in Example~\ref{ex:replace}.
The set $\Inst(\Replace{T})$
 contains the following in addition to
 $\Replace{T}$: 
\[
% \begin{array}{@{}l@{}}
%  \Inst(\Replace{T}) = \\
%  ~~
% \Replace{T} \cup
\left \{
 \begin{array}{lll}
  \TwC{\square}{\top}, 
   & 
   \TwC{\blacksquare}{\top},
   &
   \TwC{\mathsf{s}(\square)}{\top},
   \\
  \TwC{\mathsf{g}(\square,\square)}{\neg(\posvar{1}\leq \mathsf{0})},
   &
  \TwC{\mathsf{g}(\square,\blacksquare)}{\neg(\posvar{1}\leq \mathsf{0})},
  \\
  \TwC{\mathsf{g}(\mathsf{s}(\square),\square)}{\neg (\posvar{1.1} \geq \mathsf{0})},
  &
  \TwC{\mathsf{g}(\mathsf{s}(\square),\blacksquare)}{\neg (\posvar{1.1}
  \geq \mathsf{0})},
  \\
   \TwC{\mathsf{g}(\mathsf{s}(\square),\square)}{\posvar{1} \leq
   \mathsf{0} \wedge \posvar{1.1} \geq \mathsf{0}}, 
   &
   \TwC{\mathsf{g}(\mathsf{s}(\square),\square)}{\posvar{1} \leq
   \mathsf{0} \wedge \neg (\posvar{1.1} \geq \mathsf{0})}, 
   \\
  \TwC{\mathsf{g}(\mathsf{s}(\square),\square)}{\neg (\posvar{1} \leq
   \mathsf{0}) \wedge \posvar{1.1} \geq \mathsf{0}}, 
   &
   \TwC{\mathsf{g}(\mathsf{s}(\square),\square)}{\neg (\posvar{1} \leq
   \mathsf{0}) \wedge \neg (\posvar{1.1} \geq \mathsf{0})},
   \\
   \TwC{\mathsf{g}(\mathsf{s}(\square),\blacksquare)}{\posvar{1} \leq
   \mathsf{0} \wedge \posvar{1.1} \geq \mathsf{0}},
   &
  \TwC{\mathsf{g}(\mathsf{s}(\square),\blacksquare)}{\posvar{1} \leq
   \mathsf{0} \wedge \neg (\posvar{1.1} \geq \mathsf{0})},
   \\
   \TwC{\mathsf{g}(\mathsf{s}(\square),\blacksquare)}{\neg (\posvar{1}
   \leq \mathsf{0}) \wedge \posvar{1.1} \geq \mathsf{0}}, 
   &  
   \TwC{\mathsf{g}(\mathsf{s}(\square),\blacksquare)}{\neg (\posvar{1}
   \leq \mathsf{0}) \wedge \neg (\posvar{1.1} \geq \mathsf{0})}
   \\
 \end{array}
\right\}
% \end{array}
\]
The relation between the constrained patterns in $\Inst(\Replace{T})$
 w.r.t.\ $\sqsubseteq$ is illustrated as follows:
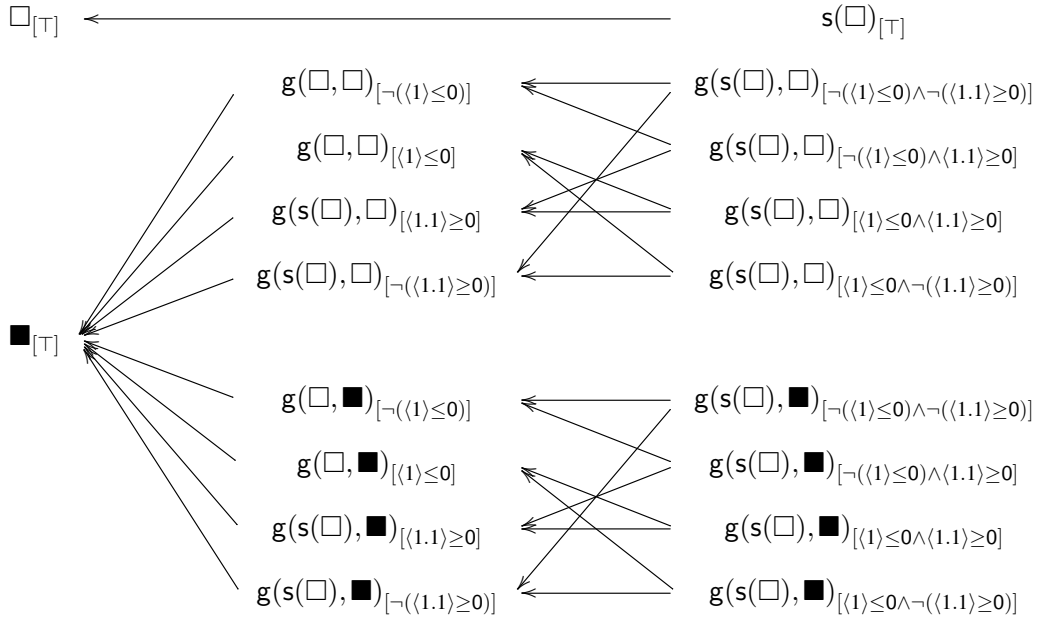
\[
 \xymatrix@C=0pt@R=5pt{
%%%%
 \TwC{\square}{\top} 
 & & \hspace{50pt}
 & & & & \hspace{50pt} &  \ar[llllll]
 & \TwC{\mathsf{s}(\square)}{\top}
 \\
%%%%
 & & & \ar[ddddll]
 & \TwC{\mathsf{g}(\square,\square)}{\neg(\posvar{1}\leq \mathsf{0})}
 & & & \ar[ll] \ar[llddd]
 & \TwC{\mathsf{g}(\mathsf{s}(\square),\square)}
       {\neg (\posvar{1} \leq \mathsf{0})
        \wedge \neg (\posvar{1.1} \geq \mathsf{0})}
 \\
 & & & \ar[dddll]
 & \TwC{\mathsf{g}(\square,\square)}{\posvar{1}\leq \mathsf{0}}
 & & & \ar[ull] \ar[dll]
 & \TwC{\mathsf{g}(\mathsf{s}(\square),\square)}
       {\neg (\posvar{1} \leq \mathsf{0}) 
        \wedge \posvar{1.1} \geq \mathsf{0}}
 \\
 & & & \ar[ddll]
 & \TwC{\mathsf{g}(\mathsf{s}(\square),\square)}{\posvar{1.1} \geq
 \mathsf{0}}
 & & & \ar[ull] \ar[ll]
 & \TwC{\mathsf{g}(\mathsf{s}(\square),\square)}
       {\posvar{1} \leq \mathsf{0} 
        \wedge \posvar{1.1} \geq \mathsf{0}}
\\
 & & & \ar[dll]
 & \TwC{\mathsf{g}(\mathsf{s}(\square),\square)}{\neg(\posvar{1.1} \geq \mathsf{0})}
 & & & \ar[uull] \ar[ll]
 & \TwC{\mathsf{g}(\mathsf{s}(\square),\square)}
       {\posvar{1} \leq \mathsf{0} 
        \wedge \neg (\posvar{1.1} \geq \mathsf{0})}
 \\
 \TwC{\blacksquare}{\top}
 & 
 \\
 & & & \ar[ull]
 & \TwC{\mathsf{g}(\square,\blacksquare)}{\neg(\posvar{1} \leq \mathsf{0})}
 & & & \ar[dddll] \ar[ll]
 & \TwC{\mathsf{g}(\mathsf{s}(\square),\blacksquare)}
 {\neg (\posvar{1} \leq \mathsf{0})
 \wedge \neg (\posvar{1.1} \geq \mathsf{0})}
 \\
 & & & \ar[uull]
 & \TwC{\mathsf{g}(\square,\blacksquare)}{\posvar{1} \leq \mathsf{0}}
 & & & \ar[ull] \ar[dll]
 & \TwC{\mathsf{g}(\mathsf{s}(\square),\blacksquare)}
       {\neg (\posvar{1} \leq \mathsf{0}) 
        \wedge \posvar{1.1} \geq \mathsf{0}}
 \\
 & & & \ar[uuull]
 & \TwC{\mathsf{g}(\mathsf{s}(\square),\blacksquare)}{\posvar{1.1} \geq
 \mathsf{0}} 
 & & & \ar[ull] \ar[ll]
 & \TwC{\mathsf{g}(\mathsf{s}(\square),\blacksquare)}
       {\posvar{1} \leq \mathsf{0} 
        \wedge \posvar{1.1} \geq \mathsf{0}}
 \\
 & & & \ar[uuuull]
 & \TwC{\mathsf{g}(\mathsf{s}(\square),\blacksquare)}{\neg(\posvar{1.1}
 \geq \mathsf{0})} 
 & & & \ar[uull] \ar[ll]
 & \TwC{\mathsf{g}(\mathsf{s}(\square),\blacksquare)}
       {\posvar{1} \leq \mathsf{0}
        \wedge \neg (\posvar{1.1} \geq \mathsf{0})}
\\
%%%%
 }
\]
where $\longleftarrow$ denotes $\sqsubset$.
\end{example}

Finally, we show a construction of a CTA recognizing $G(T)$.
We prepare two kinds of states $q_u$ and $\tilde{q}_u$
for each pattern to distinguish whether the term with the state
satisfies the constraint in the corresponding constrained term.
In the following, we use $\dot{q}$ to denote $q$ or $\tilde{q}$, and,
given an unconstrained pattern $u$ and a set $U$ of constrained
patterns, we denote the set of \emph{least general} constrained patterns
in $U$ which are more general than $u$ by $\Maximal[u](U)$:
\[
 \Maximal[u](U) = \{ \TwC{u'}{\phi'} \in U \mid 
 u' \sqsubseteq u,
 ~
 (\nexists \TwC{u''}{\phi''} \in U.\ 
 u'' \sqsubseteq u \,\wedge\,
 \TwC{u'}{\phi'} \sqsubset \TwC{u''}{\phi''}) \}
\]
Now we give an intuitive outline of the construction.
\begin{description}
 \item[Final states]
	    For a pattern $u$ capturing an instance of $\TwC{t}{\phi}$
	    in $T$,
	    we add the state $\tilde{q}_u$ to $Q_f$ if $u$ also captures
	    an instance of a non-variable proper subterm of
	    $\TwC{t'}{\phi'}$ in $T$;
	    otherwise, we add $\tilde{q}_\square$ and
	    $\tilde{q}_\blacksquare$ where $u$ is in
	    $\Terms_\square$ and
	    $\Terms_\blacksquare$, resp.
 \item[States]
	    In addition to $Q_f$, for a pattern $u$ capturing an
	    instance of a non-variable proper subterm of
	    $\TwC{t'}{\phi'}$ in $T$, we add $q_u$ to $Q$;
	    for arbitrary terms in $\Terms_\square$ and
	    $\Terms_\blacksquare$, we add $q_\square$ and
	    $q_\blacksquare$ to $Q$, resp.
 \item[Transition rules]
	    $f(\dot{q}_{u_1},\ldots,\dot{q}_{u_n})$ transitions to 
	    a state by considering the following properties of the least
	    general pattern $u$ which is more general than
	    $f(u_1,\ldots,u_n)$: 
	    \begin{enumerate}
	     \renewcommand{\labelenumi}{(\alph{enumi})}
	     \renewcommand{\theenumi}{(\alph{enumi})}
	     \item\label{condition-1}
		  whether the pattern $u$ also matches a non-variable
		  \emph{proper} subterm of some initial constrained term 
		  $\TwC{t}{\phi}$ $\in$ $T$,
	     \item\label{condition-2}
		  whether the pattern $u$ also matches an instance of
		  $\TwC{t'}{\phi'}$ in $T$, and
	     \item\label{condition-3} 
		  whether $u$ is in $\Terms_\square$. 
	    \end{enumerate}
	    The state $f(\dot{q}_{u_1},\ldots,\dot{q}_{u_n})$
	    transitions to is decided as follows:
	    \begin{center}
	     \begin{tabular}{c|ccc|c}
	      state
	      & \ref{condition-1} 
	      & \ref{condition-2} 
	      & \ref{condition-3}
	      & the transition rule is in \\
	      \hline
%	      \raisebox{0pt}[12pt]{}
	      $q_u$ & yes & no & --- & $\Delta_{\mathsf{pat}}$ \\
%	      \raisebox{0pt}[12pt]{}
	      $\tilde{q}_u$ & yes & yes & --- & $\tilde{\Delta}_{\mathsf{pat}}$ \\
%	      \raisebox{0pt}[12pt]{}
	      $q_\square$ & no & no & yes & $\Delta_\square$ \\
%	      \raisebox{0pt}[12pt]{}
	      $\tilde{q}_\square$ & no & yes & yes &
	      $\tilde{\Delta}_\square$ \\
%	      \raisebox{0pt}[12pt]{}
	      $q_\blacksquare$ & no & no & no & $\Delta_\blacksquare$ \\
%	      \raisebox{0pt}[12pt]{}
	      $\tilde{q}_\blacksquare$ & no & yes & no &
	      $\tilde{\Delta}_\blacksquare$ \\
	     \end{tabular}
	    \end{center}
\end{description}
This approach to the construction is formalized as follows.
\begin{definition}
 \label{def:G(T)}
 For a finite set $T$ of constrained terms, 
 we prepare the set $\Labels(T)$ of constrained patterns whose term parts
 are used as labels for states:
 \[
 \Labels(T) = \Inst(\Replace{T})
 \]
 The subset $\Initials(T)$ of constrained patterns that match
 elements of $T$ 
 is defined as follows:
 \[
 \Initials(T) = \{
 \TwC{u}{\phi} \in \Labels(T) \mid
 \exists \TwC{u'}{\phi'} \in \Replace{T}.\
 \TwC{u'}{\phi'} \sqsubseteq \TwC{u}{\phi}
 \}
 \]
We prepare the set $Q_0$ of candidates for states as follows:
\[
 Q_0 = \{ q_u \mid \exists \phi.\ 
 \TwC{u}{\phi} \in \Labels(T)
 \,\wedge\,
  (\exists u' \in \Subpatterns(\Replace{T}).\ 
  u' \sqsubseteq u
  )
 \}
\]
 Then, we define a CTA $\cA$ $=$ $(Q,Q_f,\Delta)$ such that
\[
 \begin{array}{@{}r@{~}r@{~}l@{}}
  Q_f & = & \{ \tilde{q}_u \mid 
   q_u \in Q_0,
   ~
   \exists \phi.\ \TwC{u}{\phi} \in \Initials(T)
   \}
   \\
  & & \cup
   \{ \tilde{q}_\square \mid 
   \exists \TwC{u}{\phi} \in
   \Replace{T}.\ 
   u \in \Terms_\square 
   \,\wedge\,
   q_u \not\in Q_0
%   (\nexists u' \in \Subpatterns(\Replace{T}).\ 
%   u' \sqsubseteq u
%   )
   \} 
   \cup
   \{ \tilde{q}_\blacksquare \mid 
   \exists \TwC{u}{\phi} \in
   \Replace{T}.\ 
   u \in \Terms_\blacksquare 
   \,\wedge\,
   q_u \not\in Q_0
%   (\nexists u' \in \Subpatterns(\Replace{T}).\
%   u' \sqsubseteq u
%   )
   \} 
   \\[5pt]
  Q & = & 
  Q_f 
  \cup
  \{ q_u \in Q_0 \mid \tilde{q}_u \not\in Q_f \}
  \cup 
  \{ q_\square, q_\blacksquare\}
  \\[3pt]
  \Delta & = & \Delta_{\mathsf{pat}} \cup \tilde{\Delta}_{\mathsf{pat}} 
  \cup \Delta_\square \cup \tilde{\Delta}_\square
  \cup \Delta_\blacksquare \cup \tilde{\Delta}_\blacksquare
 \end{array}
\]
where
\[
 \begin{array}{@{}l@{~}l@{~}l@{~}l@{~}c@{~}c@{~}c@{~}c@{~}c@{}r@{}}
%%%%%
  \multicolumn{2}{@{}l@{}}{\Delta_{\mathsf{pat}} := }
  & & &
   \overbrace{~~~~~~~~~~~}^{\mbox{\rm \ref{condition-1}}} & &
   \overbrace{~~~~~~~~~~~~~~~~~~~~~~~~~}^{\mbox{\rm \ref{condition-2}}} & & 
   \overbrace{~~~~~~~~~~}^{\mbox{\rm \ref{condition-3}}} 
   \\
  &
   \{ f(\dot{q}_{u_1},\ldots,\dot{q}_{u_n}) \xrightarrow[]{\phi} q_u
   \mid &
   \multicolumn{2}{@{}l@{}}{
   f \in \cF\cup\cG,
   ~
   \dot{q}_{u_1},\ldots,\dot{q}_{u_n} \in Q, 
   ~
   q_u \in Q,
   }
   \\
  & &
   & \TwC{u}{\phi} \in
   \Maximal[f(u_1,\ldots,u_n)](\Labels(T)),
   &
%   \tilde{q}_u \not\in Q_f,
   &  & 
   \TwC{u}{\phi} \not\in \Initials(T) 
   & & 
   & \}
   \\
%%%%%
  \lefteqn{\tilde{\Delta}_{\mathsf{pat}} = }
  \\
  &
   \{ f(\dot{q}_{u_1},\ldots,\dot{q}_{u_n}) \xrightarrow[]{\phi} \tilde{q}_u
   \mid &
   \multicolumn{5}{@{}l@{}}{
   f \in \cF\cup\cG,
   ~
   \dot{q}_{u_1},\ldots,\dot{q}_{u_n} \in Q, 
   ~
   \tilde{q}_u \in Q,
   }
   &
   \\
  & &
   & \TwC{u}{\phi} \in
   \Maximal[f(u_1,\ldots,u_n)](\Labels(T)),
   &
%   \tilde{q}_u \in Q_f,
   &  & 
   \TwC{u}{\phi} \in \Initials(T) 
   & & 
   & \}
   \\
%%%%%
  \lefteqn{\Delta_\square = }
  \\
  &
   \{ f(\dot{q}_{u_1},\ldots,\dot{q}_{u_n}) \xrightarrow[]{\phi} q_\square
   \mid &
   \multicolumn{5}{@{}l@{}}{
   f \in \cF\cup\cG,
   ~
   \dot{q}_{u_1},\ldots,\dot{q}_{u_n} \in Q, 
   ~
   q_\square \in Q, } 
   &
   \\
  & &
   (\exists & \TwC{u}{\phi} \in
   \Maximal[f(u_1,\ldots,u_n)](\Labels(T)).
   &
   q_u \not\in Q_0
   & \wedge & 
   \TwC{u}{\phi} \not\in \Initials(T) 
   & \wedge & 
   u \in \Terms_\square 
   & ) \}
   \\
%%%%%
  \lefteqn{\tilde{\Delta}_\square = }
  \\
  &
   \{ f(\dot{q}_{u_1},\ldots,\dot{q}_{u_n}) \xrightarrow[]{\phi} \tilde{q}_\square
   \mid &
   \multicolumn{5}{@{}l@{}}{
   f \in \cF\cup\cG,
   ~
   \dot{q}_{u_1},\ldots,\dot{q}_{u_n} \in Q, 
   ~
   \tilde{q}_\square \in Q, } 
   &
   \\
  & &
   (\exists & \TwC{u}{\phi} \in
   \Maximal[f(u_1,\ldots,u_n)](\Labels(T)).
   &
   q_u \not\in Q_0
   & \wedge & 
   \TwC{u}{\phi} \in \Initials(T) 
   & \wedge & 
   u \in \Terms_\square 
   & ) \}
   \\
%%%%%
  \lefteqn{\Delta_\blacksquare = }
  \\
  &
   \{ f(\dot{q}_{u_1},\ldots,\dot{q}_{u_n}) \xrightarrow[]{\phi} q_\blacksquare
   \mid &
   \multicolumn{5}{@{}l@{}}{
   f \in \cF\cup\cG,
   ~
   \dot{q}_{u_1},\ldots,\dot{q}_{u_n} \in Q, 
   ~
   q_\blacksquare \in Q, } 
   &
   \\
  & &
   (\exists & \TwC{u}{\phi} \in
   \Maximal[f(u_1,\ldots,u_n)](\Labels(T)).
   &
   q_u \not\in Q_0
   & \wedge & 
   \TwC{u}{\phi} \not\in \Initials(T) 
   & \wedge & 
   u \in \Terms_\blacksquare 
   & ) \}
   \\
%%%%%
  \lefteqn{\tilde{\Delta}_\blacksquare = }
  \\
  &
   \{ f(\dot{q}_{u_1},\ldots,\dot{q}_{u_n}) \xrightarrow[]{\phi} \tilde{q}_\blacksquare
   \mid &
   \multicolumn{5}{@{}l@{}}{
   f \in \cF\cup\cG,
   ~
   \dot{q}_{u_1},\ldots,\dot{q}_{u_n} \in Q, 
   ~
   \tilde{q}_\blacksquare \in Q, } 
   &
   \\
  & &
   (\exists & \TwC{u}{\phi} \in
   \Maximal[f(u_1,\ldots,u_n)](\Labels(T)).
   &
   q_u \not\in Q_0
   & \wedge & 
   \TwC{u}{\phi} \in \Initials(T) 
   & \wedge & 
   u \in \Terms_\blacksquare 
   & ) \}
   \\
%%%%%
 \end{array}
\]
\end{definition}
Note that for any pattern $u$ $\not\in$ $\{\square,\blacksquare\}$, 
% (i) if $\tilde{q}_u$ $\in$ $Q_f$, then $q_u$ $\not\in$ $Q$, and
%(ii)
$q_u$ $\in$ $Q_0$ iff $q_u$ $\in$ $Q$ or $\tilde{q}_u$ $\in$ $Q_f$, and
thus, $q_u$ $\in$ $Q_0$ which makes \ref{condition-1} true implicitly
holds in the conditions of $\Delta_{\mathsf{pat}}$ and
$\tilde{\Delta}_{\mathsf{pat}}$.
Note also that the constructed transition rules are not always
optimized via the construction in Definition~\ref{def:G(T)}.%
\footnote{ As in the case of unconstrained tree automata~\cite{TATA}, 
a state $q$ is not accessible if there is no transition rule of the
form $l \xrightarrow[]{\phi} q$.
Thus, such a state $q$ and all the transition rules having $q$ can be
removed from $\cA$.}
\begin{theorem}
 \label{th:correctness}
The CTA $\cA$ constructed in Definition~\ref{def:G(T)} is a
 deterministic, complete, and constraint-complete CTA such that $L(\cA)$
 $=$ $G(T)$.
\end{theorem}
\begin{proof}
Due to the construction of constrained patterns used for states, 
for $\TwC{u}{\phi}$ $\in$ $\Labels(T)$, 
all of the following hold:
\begin{itemize}
 \item $\pi$ $\in$ $\Pos(u)$ and $u|_\pi$ $\in$ $\Terms_\square$
       for any variable $\posvar{\pi}$ $\in$ $\FVar(\phi)$, 
 \item $\bigvee_{\TwC{u}{\psi} \in \Labels(T)} \psi$ is
       valid w.r.t.\ $\cS$, and
 \item if $\TwC{u}{\phi}$ is a least general in $\Labels(T)$, then
       $\phi \wedge \psi$ is not satisfiable w.r.t.\ $\cS$ for any
       least general constrained pattern $\TwC{u}{\psi}$ in $\Labels(T)$
       such that $\phi$ $\not\bumpeq$ $\psi$.
 \end{itemize}
It follows from both the first property above and the construction of
 transition rules that $\cA$ is constraint-complete.  
It follows from the second and third properties above that $\cA$ is
 complete and deterministic. 
From these properties, every given term transitions to a unique state
 that keeps the structure of the given term as much as possible in the
 sense of the patterns to be considered. 
 In constructing transition rules, we add constraints to transition
 rules so as to transition to a final state if a given initial ground
 term is an instance of a constrained term in $T$. 
 For these reasons, it is clear that $L(\cA)$ $=$ $G(T)$.
\end{proof}

\begin{example} 
 \label{ex:complete}
Consider $\cG$, $\cP$ and $\cS$ in Example~\ref{ex:replace},
$\cF$ $=$ $\{ \mathsf{f}/1 \}$, and
$T$ $=$ $\{ \TwC{\mathsf{f}(x)}{x \leq \mathsf{0}}, 
        \TwC{\mathsf{f}(\mathsf{s}(x))}{x \geq \mathsf{0}} \}$.
Then, we have that 
\[
 \begin{array}{@{}l@{}}
  \Replace{T}
   =
   \{ \TwC{\mathsf{f}(\square)}{\posvar{1}\leq \mathsf{0}},
   \TwC{\mathsf{f}(\mathsf{s}(\square))}{\posvar{1.1}\geq
   \mathsf{0}}\}\\[5pt]
   \Subpatterns(\Replace{T})
   =
   \{ \mathsf{s}(\square) \} \\[5pt]
   \Inst(\Replace{T})
   =
   \Replace{T} \cup 
   \{ \TwC{\square}{\top}, \TwC{\mathsf{s}(\square)}{\top} \} \\
  ~\hspace{26.5mm}
   \cup
   \left\{
    \begin{array}{@{\>}l@{\>}}
     \TwC{\blacksquare}{\top},
   \TwC{\mathsf{f}(\square)}{\neg (\posvar{1}\leq\mathsf{0})},
   \TwC{\mathsf{f}(\mathsf{s}(\square))}{\neg
   (\posvar{1.1}\geq\mathsf{0})},
   \TwC{\mathsf{f}(\mathsf{s}(\square))}{\posvar{1}\leq\mathsf{0}\wedge\posvar{1.1}\geq\mathsf{0}}, \\
%  ~\hspace{32mm}
   \TwC{\mathsf{f}(\mathsf{s}(\square))}{\neg (\posvar{1}\leq\mathsf{0})\wedge\posvar{1.1}\geq\mathsf{0}},
   \TwC{\mathsf{f}(\mathsf{s}(\square))}{\posvar{1}\leq\mathsf{0}\wedge\neg
   (\posvar{1.1}\geq\mathsf{0})},
   \TwC{\mathsf{f}(\mathsf{s}(\square))}{\neg (\posvar{1}\leq\mathsf{0})
   \wedge \neg (\posvar{1.1}\geq\mathsf{0})}
   \\
   \end{array}
   \right\} \\[12pt]
   \Initials(T)
   =
   \Replace{T}
   \cup
   \{
   \TwC{\mathsf{f}(\mathsf{s}(\square))}{\posvar{1}\leq\mathsf{0}\wedge\posvar{1.1}\geq\mathsf{0}},
       \TwC{\mathsf{f}(\mathsf{s}(\square))}{\posvar{1}\leq\mathsf{0}\wedge\neg
       (\posvar{1.1}\geq\mathsf{0})},
       \TwC{\mathsf{f}(\mathsf{s}(\square))}{\neg
       (\posvar{1}\leq\mathsf{0}) \wedge \posvar{1.1}\geq\mathsf{0}}
       \} \\
 \end{array}
\]
The CTA $\cA$ $=$ $(
 \{ 
 q_\square, q_\blacksquare,  q_{\mathsf{s}(\square)},
 \tilde{q}_{\blacksquare}
 \},
 \{ 
 \tilde{q}_{\blacksquare}
 \},
 \Delta)$ is constructed by Definition~\ref{def:G(T)} with the following
 transition rules:
 \[
 \Delta =
 \left \{
 \begin{array}{l@{~~~~}l@{~~~~}l@{~~~~}l}
  \mathsf{0} \to q_\square,
   &
  \mathsf{s}(q_\square) \to q_{\mathsf{s}(\square)},
   &
   \mathsf{f}(q_\square) \xrightarrow[]{\posvar{1} \leq \mathsf{0}} 
   \tilde{q}_\blacksquare,
   &
   \mathsf{f}(q_{\mathsf{s}(\square)}) \xrightarrow[]{\posvar{1}
   \leq \mathsf{0} \wedge \posvar{1.1} \geq \mathsf{0}} 
   \tilde{q}_\blacksquare,
   \\
  & 
  \mathsf{s}(q_\blacksquare) \to q_\blacksquare,
   &
   \mathsf{f}(q_\square) \xrightarrow[]{\neg (\posvar{1} \leq \mathsf{0})} 
   q_{\blacksquare},
   &
   \mathsf{f}(q_{\mathsf{s}(\square)}) \xrightarrow[]{\posvar{1}
   \leq \mathsf{0} \wedge \neg (\posvar{1.1} \geq \mathsf{0})} 
   \tilde{q}_\blacksquare,
   \\
  &
  \mathsf{s}(q_{\mathsf{s}(\square)}) \to q_{\mathsf{s}(\square)},
   &
   \mathsf{f}(q_\blacksquare) \to q_\blacksquare,
   &
   \mathsf{f}(q_{\mathsf{s}(\square)}) \xrightarrow[]{\neg
   (\posvar{1} \leq \mathsf{0}) \wedge \posvar{1.1} \geq \mathsf{0}} 
   \tilde{q}_\blacksquare,
   \\
  &
   \mathsf{s}(\tilde{q}_\blacksquare) 
   \to q_\blacksquare,
   &
   \mathsf{f}(\tilde{q}_\blacksquare) \to q_\blacksquare,
   &
   \mathsf{f}(q_{\mathsf{s}(\square)}) \xrightarrow[]{\neg
   (\posvar{1} \leq \mathsf{0}) \wedge \neg (\posvar{1.1} \geq \mathsf{0})} 
   q_{\blacksquare}
   \\ 
 \end{array}
 \right \}
 \]
\end{example}

%------------------------------------------------------------------------------
\section{Conclusion}
\label{sec:conclusion}

In this paper, we proposed a construction method of deterministic,
complete, and constraint-complete CTAs recognizing ground instances of
constrained terms. 
For the lack of space, we did not describe how to apply it to the
verification of reduction-completeness and sufficient completeness,
while we have already worked for some examples. 

Unlike the case of tree automata, for a state, it is in general
undecidable whether there exists a term reachable to the state, and
thus, the intersection emptiness problem of CTAs is
undecidable in general (see the case of
AWEDC~\cite[Theorem~4.2.10]{TATA}).
For this reason, we will use a trivial sufficient condition that the set
of final states of product automata is empty. 
Surprisingly, this is sometimes useful for product automata. 
To make this approach more powerful, we need to develop a method to
find states that are not reachable from any ground term, e.g., there is
a possibility to detect a transition rule that is never used:
for $\mathsf{f}(q_{\mathsf{s}(\square)}) \xrightarrow[]{\neg (\posvar{1}
\leq \mathsf{0}) \wedge \neg (\posvar{1.1} \geq \mathsf{0})} q_{\blacksquare}$ in
	 Example~\ref{ex:complete},
we know that in applying this rule, the first argument of $\mathsf{f}$
is always an interpretable term of the form $\mathsf{s}(t)$, and thus,
$\posvar{1}$ in the constraint can be replaced by
$\mathsf{s}(\posvar{1.1})$;
then we can notice that the constraint is unsatisfiable, and thus, this
transition rule is never used. 
Formalizing this observation is one of our future work. 

%------------------------------------------------------------------------------
% Refs:
%
\bibliographystyle{eptcs}
%\bibliography{ref}

\begin{thebibliography}{10}
\providecommand{\bibitemdeclare}[2]{}
\providecommand{\surnamestart}{}
\providecommand{\surnameend}{}
\providecommand{\urlprefix}{Available at }
\providecommand{\url}[1]{\texttt{#1}}
\providecommand{\href}[2]{\texttt{#2}}
\providecommand{\urlalt}[2]{\href{#1}{#2}}
\providecommand{\doi}[1]{doi:\urlalt{http://dx.doi.org/#1}{#1}}
\providecommand{\bibinfo}[2]{#2}

\bibitemdeclare{inproceedings}{BJ08}
\bibitem{BJ08}
\bibinfo{author}{Adel \surnamestart Bouhoula\surnameend} \&
  \bibinfo{author}{Florent \surnamestart Jacquemard\surnameend}
  (\bibinfo{year}{2008}): \emph{\bibinfo{title}{Automated Induction with
  Constrained Tree Automata}}.
\newblock In \bibinfo{editor}{Alessandro \surnamestart Armando\surnameend},
  \bibinfo{editor}{Peter \surnamestart Baumgartner\surnameend} \&
  \bibinfo{editor}{Gilles \surnamestart Dowek\surnameend}, editors: {\sl
  \bibinfo{booktitle}{Proceedings of the 4th International Joint Conference on
  Automated Reasoning}}, {\sl \bibinfo{series}{Lecture Notes in
  Computer Science}} \bibinfo{volume}{5195}, \bibinfo{publisher}{Springer}, pp.
  \bibinfo{pages}{539--554}, \doi{10.1007/978-3-540-71070-7\_44}.

\bibitemdeclare{article}{BJ12}
\bibitem{BJ12}
\bibinfo{author}{Adel \surnamestart Bouhoula\surnameend} \&
  \bibinfo{author}{Florent \surnamestart Jacquemard\surnameend}
  (\bibinfo{year}{2012}): \emph{\bibinfo{title}{Sufficient completeness
  verification for conditional and constrained {TRS}}}.
\newblock {\sl \bibinfo{journal}{Journal of Applied Logic}}
  \bibinfo{volume}{10}(\bibinfo{number}{1}), pp. \bibinfo{pages}{127--143},
  \doi{10.1016/j.jal.2011.09.001}.

\bibitemdeclare{misc}{TATA}
\bibitem{TATA}
\bibinfo{author}{Hubert \surnamestart Comon\surnameend}, \bibinfo{author}{Max
  \surnamestart Dauchet\surnameend}, \bibinfo{author}{R\'emi \surnamestart
  Gilleron\surnameend}, \bibinfo{author}{Florent \surnamestart
  Jacquemard\surnameend}, \bibinfo{author}{Denis \surnamestart
  Lugiez\surnameend}, \bibinfo{author}{Christof \surnamestart
  L\"oding\surnameend}, \bibinfo{author}{Sophie \surnamestart Tison\surnameend}
  \& \bibinfo{author}{Marc \surnamestart Tommasi\surnameend}
  (\bibinfo{year}{2007}): \emph{\bibinfo{title}{Tree Automata Techniques and
  Applications}}.
\newblock \bibinfo{howpublished}{Available on:
  \url{http://www.grappa.univ-lille3.fr/tata}}.
\newblock \bibinfo{note}{Release October, 12th 2007}.

\bibitemdeclare{article}{Fri89}
\bibitem{Fri89}
\bibinfo{author}{Laurent \surnamestart Fribourg\surnameend}
  (\bibinfo{year}{1989}): \emph{\bibinfo{title}{A Strong Restriction of the
  Inductive Completion Procedure}}.
\newblock {\sl \bibinfo{journal}{Journal of Symbolic Computation}}
  \bibinfo{volume}{8}(\bibinfo{number}{3}), pp. \bibinfo{pages}{253--276},
  \doi{10.1016/S0747-7171(89)80069-0}.

\bibitemdeclare{article}{furuichi}
\bibitem{furuichi}
\bibinfo{author}{Yuki \surnamestart Furuichi\surnameend},
  \bibinfo{author}{Naoki \surnamestart Nishida\surnameend},
  \bibinfo{author}{Masahiko \surnamestart Sakai\surnameend},
  \bibinfo{author}{Keiichirou \surnamestart Kusakari\surnameend} \&
  \bibinfo{author}{Toshiki \surnamestart Sakabe\surnameend}
  (\bibinfo{year}{2008}): \emph{\bibinfo{title}{Approach to Procedural-program
  Verification Based on Implicit Induction of Constrained Term Rewriting
  Systems}}.
\newblock {\sl \bibinfo{journal}{IPSJ Transactions on Programming}}
  \bibinfo{volume}{1}(\bibinfo{number}{2}), pp. \bibinfo{pages}{100--121}.
\newblock \bibinfo{note}{In Japanese}.

\bibitemdeclare{book}{Logic}
\bibitem{Logic}
\bibinfo{author}{Michael \surnamestart Huth\surnameend} \&
  \bibinfo{author}{Mark \surnamestart Ryan\surnameend} (\bibinfo{year}{2000}):
  \emph{\bibinfo{title}{Logic in Computer Science: Modelling and Reasoning
  about Systems}}.
\newblock \bibinfo{publisher}{Cambridge University Press}.

\bibitemdeclare{article}{NNKSS11}
\bibitem{NNKSS11}
\bibinfo{author}{Naoki \surnamestart Nakabayashi\surnameend},
  \bibinfo{author}{Naoki \surnamestart Nishida\surnameend},
  \bibinfo{author}{Keiichirou \surnamestart Kusakari\surnameend},
  \bibinfo{author}{Toshiki \surnamestart Sakabe\surnameend} \&
  \bibinfo{author}{Masahiko \surnamestart Sakai\surnameend}
  (\bibinfo{year}{2011}): \emph{\bibinfo{title}{Lemma Generation Method in
  Rewriting Induction for Constrained Term Rewriting Systems}}.
\newblock {\sl \bibinfo{journal}{Computer Software}}
  \bibinfo{volume}{28}(\bibinfo{number}{1}), pp. \bibinfo{pages}{173--189},
  \doi{10.11309/jssst.28.1\_173}.
\newblock \bibinfo{note}{In Japanese}.

\bibitemdeclare{inproceedings}{CTA}
\bibitem{CTA}
\bibinfo{author}{Naoki \surnamestart Nishida\surnameend},
  \bibinfo{author}{Futoshi \surnamestart Nomura\surnameend},
  \bibinfo{author}{Katsuhisa \surnamestart Kurahashi\surnameend} \&
  \bibinfo{author}{Masahiko \surnamestart Sakai\surnameend}
  (\bibinfo{year}{2012}): \emph{\bibinfo{title}{Constrained Tree Automata and
  their Closure Properties}}.
\newblock In \bibinfo{editor}{Keisuke \surnamestart Nakano\surnameend} \&
  \bibinfo{editor}{Hiroyuki \surnamestart Seki\surnameend}, editors: {\sl
  \bibinfo{booktitle}{Proceedings of the 1st International Workshop on Trends
  in Tree Automata and Tree Transducers}}, pp.
  \bibinfo{pages}{24--34}.

\bibitemdeclare{inproceedings}{SNS11}
\bibitem{SNS11}
\bibinfo{author}{Tsubasa \surnamestart Sakata\surnameend},
  \bibinfo{author}{Naoki \surnamestart Nishida\surnameend} \&
  \bibinfo{author}{Toshiki \surnamestart Sakabe\surnameend}
  (\bibinfo{year}{2011}): \emph{\bibinfo{title}{On Proving Termination of
  Constrained Term Rewriting Systems by Eliminating Edges from Dependency
  Graphs}}.
\newblock In \bibinfo{editor}{Herbert \surnamestart Kuchen\surnameend}, editor:
  {\sl \bibinfo{booktitle}{Proceedings of the 20th International Workshop on
  Functional and (Constraint) Logic Programming (WFLP 2011)}}, {\sl
  \bibinfo{series}{Lecture Notes in Computer Science}} \bibinfo{volume}{6816},
  \bibinfo{publisher}{Springer}, pp. \bibinfo{pages}{138--155},
  \doi{10.1007/978-3-642-22531-4\_9}.

\bibitemdeclare{article}{sakata}
\bibitem{sakata}
\bibinfo{author}{Tsubasa \surnamestart Sakata\surnameend},
  \bibinfo{author}{Naoki \surnamestart Nishida\surnameend},
  \bibinfo{author}{Toshiki \surnamestart Sakabe\surnameend},
  \bibinfo{author}{Masahiko \surnamestart Sakai\surnameend} \&
  \bibinfo{author}{Keiichirou \surnamestart Kusakari\surnameend}
  (\bibinfo{year}{2009}): \emph{\bibinfo{title}{Rewriting Induction for
  Constrained Term Rewriting Systems}}.
\newblock {\sl \bibinfo{journal}{IPSJ Transactions on Programming}}
  \bibinfo{volume}{2}(\bibinfo{number}{2}), pp. \bibinfo{pages}{80--96}.
\newblock \bibinfo{note}{In Japanese}.

\end{thebibliography}

%------------------------------------------------------------------------------
\end{document}